\documentclass[12pt,a4paper]{amsart}

\usepackage[hmargin=2cm,vmargin=2cm]{geometry}

\usepackage{amsmath,amsfonts,amssymb}

\newcommand{\mbZ}{\mathbb Z}
\newcommand{\mbC}{\mathbb C}

\def\CP1{\mathbb{C}\mathrm{P}^1}
\newcommand{\mbP}{\mathbb P}

\newcommand{\oM}{\overline{\mathcal M}}
\newcommand{\of}{\overline f}
\newcommand{\oh}{{\overline h}}
\newcommand{\og}{\overline g}

\newcommand{\eps}{\varepsilon}
\def\d{\partial}
\def\un{{1\!\! 1}}

\newcommand{\M}{\mathcal M}

\newcommand{\cA}{{\mathcal A}}
\newcommand{\cB}{\mathcal B}

\newcommand{\<}{\left <}
\renewcommand{\>}{\right >}

\newcommand{\hcA}{{\widehat{\mathcal A}}}
\newcommand{\hcB}{{\widehat{\mathcal B}}}
\newcommand{\hLambda}{{\widehat\Lambda}}
\newcommand{\hT}{\widehat T}
\newcommand{\hQ}{\widehat Q}
\newcommand{\hZ}{\widehat Z}
\newcommand{\im}{\mathop{\mathrm{im}}\nolimits}
\renewcommand{\deg}{\mathop{\mathrm{deg}}\nolimits}
\newcommand{\Coef}{\mathop{\mathrm{Coef}}\nolimits}

\newcommand{\Mat}{\mathop{\mathrm{Mat}}\nolimits}

\def\virt{{\rm virt}}
\def\DR{{\rm DR}}
\def\LM{{\rm LM}}

\newcommand*{\der}[2]{\frac{\d #1}{\d #2}}

\newtheorem{theorem}{Theorem}[section]

\newtheorem{lemma}[theorem]{Lemma}

\theoremstyle{definition}

\newtheorem{definition}[theorem]{Definition}

\newtheorem{remark}[theorem]{Remark}

\numberwithin{equation}{section}

\title{Double ramification cycles and integrable hierarchies}

\author{A. Buryak}

\address{A.~Buryak:\newline
Department of Mathematics,
ETH Zurich,\newline
Ramistrasse 101 8092, HG G27.1, Zurich, Switzerland.}
\email{buryaksh@gmail.com}

\begin{document}

\begin{abstract}
It this paper we present a new construction of a hamiltonian hierarchy associated to a cohomological field theory. We conjecture that in the semisimple case our hierarchy is related to the Dubrovin-Zhang hierarchy by a Miura transformation and check it in several examples.
\end{abstract}

\maketitle

\section{Introduction}

In the last two decades there was a great progress in understanding relations between the topology of the moduli space of stable curves and integrable hierarchies of partial differential equations. The first result in this direction was the famous Witten's conjecture~(\cite{Wit91}) proved by Kontsevich~(\cite{Kon92}). It says that the generating series of the intersection numbers on the moduli space of stable curves is a solution of the KdV hierarchy. After that it has been expected that the Gromov-Witten invariants of any target space (or, more generally, correlators of any cohomological field theory) should be governed by an integrable hierarchy. We refer the reader to \cite{OP06,CDZ04,DZ04,MT08,CvdL13,MST14} for different results concerning the Gromov-Witten invariants of $\mbP^1$ and orbifold $\mbP^1$. The Hodge integrals on the moduli space of stable curves can be described by integrable hierarchies in two different ways~(\cite{Kaz09,Bur13}). Witten's original conjecture can be generalized to the moduli space of $r$-spin curves and its generalizations (\cite{Wit93,FSZ10,FJR13}). An integrable hierarchy corresponding to the Gromov-Witten theory of the resolved conifold is discussed in \cite{BCR12,BCRR14}. A hierarchy that governs the degree zero Gromov-Witten invariants was constructed in~\cite{Dub13}. KdV type equations for the intersection numbers on the moduli space of Riemann surfaces with boundary were introduced in \cite{PST14} and further studied in \cite{Bur14a,Bur14b}.

In \cite{DZ05} B. Dubrovin and Y. Zhang suggested a very general approach to the problem. The Gromov-Witten invariants of any target space (or, more generally, the correlators of any cohomological field theory) are packed in a generating series that is called the potential. If the potential is conformal and semisimple, then B. Dubrovin and Y. Zhang gave a construction of a bihamiltonian hierarchy of PDEs such that the potential is the logarithm of a tau-function of this hierarchy. Their construction was generalized to the non-conformal case in \cite{BPS12a} (see also~\cite{BPS12b}). 

In this paper we give a new construction of a hamiltonian hierarchy associated to a potential of Gromov-Witten type. We must immediately say that, in general, our hierarchy is different from the Dubrovin-Zhang hierarchy. Our construction is motivated by Symplectic Field Theory (see~\cite{EGH00}) and is based on the integration over the double ramification cycles. That is why we call our hierarchy the DR hierarchy. Our construction is quite different from Dubrovin and Zhang's construction and also from the construction in \cite{BPS12a}. First of all, we don't need the assumption of the semisimplicity. Second, the equations of the DR hierarchy are differential polynomials immediately by the construction, while the polynomiality of the equations of the Dubrovin-Zhang hierarchy is completely non-obvious (see~\cite{BPS12a}). Finally, the hamiltonian structures are in general very different. In the DR hierarchy it is given by the operator $\eta\d_x$, where $\eta$ is the matrix of the scalar product in a cohomological field theory. In the Dubrovin-Zhang hierarchy the hamiltonian structure can be much more complicated. 

We conjecture that, if a cohomological field theory is semisimple, then the DR hierarchy is related to the Dubrovin-Zhang hierarchy by a Miura transformation. We check it in the case of the trivial cohomological field theory and in the case of the cohomological field theory formed by the Hodge classes: $1+\eps\lambda_1+\eps^2\lambda_2+\ldots+\eps^g\lambda_g\in H^*(\oM_{g,n};\mbC)$. 


\subsection{Organization of the paper}

Section~\ref{section:algebraic preliminaries} contains the necessary algebraic formalism in the theory of formal partial differential equations that is needed for our constructions. 

In Section~\ref{section:geometric preliminaries} we recall the main geometric notions: the definition of a cohomological field theory and the definition of the double ramification cycles. 

Section~\ref{section:construction} contains the main result of the paper: the construction of the DR hierarchy. The main statement here is Theorem~\ref{theorem:main} that says that the constructed Hamiltonians commute with each other. We also formulate several general properties of the DR hierarchy, compute the hierarchy explicitly in two examples and check our conjecture in these cases. 

Section~\ref{section:commutativity} is devoted to the proof of Theorem~\ref{theorem:main}.

Appendix~\ref{appendix} contains some technical algebraic results.


\subsection{Acknowledgments}

We thank B. Dubrovin, S. Shadrin, R. Pandharipande and D. Zvonkine for discussions related to the work presented here.

We would like to thank the anonymous referee for valuable remarks and suggestions that allowed us to improve the exposition of this paper.

This work was supported by grant ERC-2012-AdG-320368-MCSK in the group of R. Pandharipande at ETH Zurich, by the Russian Federation Government grant no. 2010-220-01-077 (ag. no. 11.634.31.0005), the grants RFFI 13-01-00755 and NSh-4850.2012.1.


\section{Algebraic preliminaries: formal approach to partial differential equations}\label{section:algebraic preliminaries}

In this section we introduce the language and the formalism in the theory of formal partial differential equations that will be necessary in our construction of the DR hierarchy. The material comes almost entirely from \cite{DZ05} and \cite{Ros10}.

In Section~\ref{subsection:differential polynomials} we introduce the algebra of differential polynomials and the space of local functionals. In Section~\ref{subsection:poisson structure} we describe a certain class of Poisson structures on the space of local functionals. In Section~\ref{subsection:Poisson algebra} we introduce a Poisson algebra $\cB_N$ and discuss its relation with the space of local functionals. In Section~\ref{subsection:extended spaces} we define certain extensions of the ring of differential polynomials, the space of local functionals and the Poisson algebra $\cB_N$. 


\subsection{Differential polynomials and local functionals}\label{subsection:differential polynomials}

Here we recall the definitions of the ring of differential polynomials and the space of local functionals. 

\subsubsection{Differential polynomials}

Let us fix an integer $N\ge 1$. Consider variables $u^\alpha_j$, $1\le \alpha\le N$, $j\ge 0$. We will often denote $u^\alpha_0$ by $u^\alpha$ and use an alternative notation for the variables $u^\alpha_1,u^\alpha_2,\ldots$:
$$
u^\alpha_x:=u^\alpha_1,\quad u^\alpha_{xx}:=u^\alpha_2,\ldots.
$$
Denote by $\cA_N$ the ring of polynomials in the variables $u^\alpha_j$, $j\ge 1$,
$$
f(u;u_x,u_{xx},\ldots)=\sum_{m\ge 0}\sum_{\substack{1\le\alpha_1,\ldots,\alpha_m\le N\\j_1,\ldots,j_m\ge 1}}f^{j_1,j_2,\ldots,j_m}_{\alpha_1,\alpha_2,\ldots,\alpha_m}(u)u^{\alpha_1}_{j_1}u^{\alpha_2}_{j_2}\ldots u^{\alpha_m}_{j_m}
$$ 
with the coefficients $f^{j_1,j_2,\ldots,j_m}_{\alpha_1,\alpha_2,\ldots,\alpha_m}(u)$ being power series in $u^1,\ldots,u^N$. Elements of the ring $\cA_N$ will be called differential polynomials.

Let us introduce a gradation $\deg_{dif}$ on the ring $\cA_N$ of differential polynomials putting 
$$
\deg_{dif} u^\alpha_k=k,\, k\ge 1;\quad \deg_{dif} f(u)=0.
$$
This gradation will be called differential degree.

The operator $\d_x\colon\cA_N\to\cA_N$ is defined as follows:
$$
\d_x:=\sum_{\alpha=1}^N\sum_{s\ge 0}u^\alpha_{s+1}\frac{\d}{\d u^\alpha_s}.
$$ 

\subsubsection{Local functionals}

Let $\Lambda_N:=\cA_N/\im(\d_x)$. There is the projection $\pi\colon\cA_N\to\cA_N/\im(\d_x)$ and we will use the following notation:
$$
\int h dx:=\pi(h),
$$
for any $h\in\cA_N$. Elements of the space $\Lambda_N$ will be called local functionals. 

For a local functional $\oh=\int hdx\in\Lambda_N$, the variational derivative $\frac{\delta\oh}{\delta u^\alpha}\in\cA_N$ is defined as follows:
$$
\frac{\delta\overline h}{\delta u^\alpha}:=\sum_{i\ge 0}(-\d_x)^i\frac{\d h}{\d u^\alpha_i}.
$$ 

It is clear that the gradation $\deg_{dif}$ on $\cA_N$ induces a gradation on the space $\Lambda_N$ that will be also called differential degree.

Clearly, the derivative $\der{}{u^\alpha}\colon\cA_N\to\cA_N$ commutes with the operator $\d_x$. Therefore, the derivative $\der{}{u^\alpha}$ is correctly defined on the space of local functionals~$\Lambda_N$.


\subsection{Poisson structure on $\Lambda_N$ and hamiltonian systems of PDEs}\label{subsection:poisson structure}

In this section we introduce a certain class of Poisson structures on the space of local functionals and review the notion of a hamiltonian system of partial differential equations. 

\subsubsection{Poisson structure on $\Lambda_N$}

Let $K=(K^{\alpha,\beta})_{1\le\alpha,\beta\le N}$ be a matrix of operators 
\begin{gather}\label{eq: dif. operator}
K^{\alpha,\beta}=\sum_{j\ge 0}f^{\alpha,\beta}_j\d_x^j,
\end{gather}
where $f^{\alpha,\beta}_j\in\cA_N$ and the sum is finite. Let us define the bracket $\{\cdot,\cdot\}_K\colon\Lambda_N\times\Lambda_N\to\Lambda_N$ by
\begin{gather*}
\{\overline g,\overline h\}_K:=\int\sum_{\alpha,\beta}\frac{\delta\overline g}{\delta u^\alpha}K^{\alpha,\beta}\frac{\delta\overline h}{\delta u^\beta}dx.
\end{gather*}
The operator $K$ is called hamiltonian, if the bracket $\{\cdot,\cdot\}_K$ is antisymmetric and satisfies the Jacobi identity. It is well known that, for any symmetric matrix $\eta=(\eta^{\alpha,\beta})\in\Mat_{N,N}(\mbC)$, the operator $\eta\d_x$ is hamiltonian (see e.g.~\cite{DZ05}).

\subsubsection{Hamiltonian systems of PDEs}

A system of partial differential equations  
\begin{gather}\label{eq: system}
\frac{\d u^\alpha}{\d \tau_i}=f^\alpha_i(u;u_x,\ldots),\quad 1\le\alpha\le N,\quad i\ge 1,
\end{gather}
where $f^\alpha_i\in\cA_N$, is called hamiltonian, if there exists a hamiltonian operator $K=(K^{\alpha,\beta})$ and a sequence of local functionals $\overline h_i\in\Lambda_N$, $i\ge 1$, such that
\begin{align*}
&f^\alpha_i=\sum_\beta K^{\alpha,\beta}\frac{\delta\oh_i}{\delta u^\beta},\\
&\{\oh_i,\oh_j\}_K=0,\quad\text{for $i,j\ge 1$}.
\end{align*}
The local functionals $\oh_i$ are called the Hamiltonians of the system \eqref{eq: system}. 


\subsection{Poisson algebra $\cB_N$}\label{subsection:Poisson algebra}

Consider formal variables $p^\alpha_n$, where $1\le\alpha\le N$ and $n\in\mbZ_{\ne 0}$. Let $\cB_N\subset\mbC[[p^\alpha_n]]$ be the subalgebra that consists of power series of the form
$$
f=\sum_{k\ge 0}\sum_{\substack{1\le\alpha_1,\ldots,\alpha_k\le N\\n_1,\ldots,n_k\ne 0\\n_1+\ldots+n_k=0}}f_{\alpha_1,\ldots,\alpha_k}^{n_1,\ldots,n_k}p^{\alpha_1}_{n_1}p^{\alpha_2}_{n_2}\ldots p^{\alpha_k}_{n_k},
$$
where $f_{\alpha_1,\ldots,\alpha_k}^{n_1,\ldots,n_k}$ are complex coefficients.

Let $\eta=(\eta^{\alpha,\beta})_{1\le\alpha,\beta\le N}\in\Mat_{N,N}(\mbC)$ be a symmetric matrix. We endow the algebra $\cB_N$ with the following Poisson structure:
\begin{gather}\label{eq:poisson}
\{p^\alpha_m,p^\beta_n\}_\eta:=i m \eta^{\alpha,\beta}\delta_{m+n,0}.
\end{gather}

Let us define an important map $T_0\colon\Lambda_N\to\cB_N$. Consider an arbitrary differential polynomial $f\in\cA_N$. We want to consider the variable $u^\alpha$ as a formal function of $x$ that has a Fourier expansion with coefficients $p^\alpha_n$. Formally, we make the substitution $u^\alpha_k=\sum_{n\ne 0}(in)^k p^\alpha_n e^{inx}$. Then we have 
$$
\left.f\right|_{u^\alpha_k=\sum_{n\ne 0}(in)^k p^\alpha_n e^{inx}}=\sum_{m\in\mbZ} P_m e^{imx}, 
$$
where $P_m\in\mbC[[p^\alpha_n]]$. Denote the right-hand side by $T(f)$. Clearly, $P_0\in\cB_N$. Let $T_0(f):=P_0$. Obviously, we have $\left.T_0\right|_{\im\d_x}=0$. Therefore, the map $T_0$ induces a map $\Lambda_N=\cA_N/(\im\d_x)\to\cB_N$, that we will also denote by $T_0$. 

Denote by $\cB_N^{pol}\subset\cB_N$ the subspace that consists of power series of the form
$$
f=\sum_{k\ge 0}\sum_{\substack{1\le\alpha_1,\ldots,\alpha_k\le N\\n_1,\ldots,n_k\ne 0\\n_1+\ldots+n_k=0}}P_{\alpha_1,\ldots,\alpha_k}(n_1,\ldots,n_k)p^{\alpha_1}_{n_1}p^{\alpha_2}_{n_2}\ldots p^{\alpha_k}_{n_k},
$$
where $P_{\alpha_1,\ldots,\alpha_k}(z_1,\ldots,z_k)\in\mbC[z_1,\ldots,z_k]$ are some polynomials and the degrees of them in the power series $f$ are bounded from above. It is easy to see that $\im(T_0)\subset\cB_N^{pol}$. Let us formulate two important properties of the map~$T_0$.

\begin{lemma}\label{lemma:surjectivity}
The map $T_0\colon\Lambda_N\to\cB^{pol}_N$ is surjective. The kernel of it is $N$-dimensional and is spanned by the local functionals $\int u^\alpha dx$, $1\le\alpha\le N$.
\end{lemma}

From this lemma it follows, that for any power series $f\in\cB_N^{pol}$ there exists a unique local functional $\oh\in\Lambda_N$ such that $T_0(\oh)=f$ and $\oh$ has the form $\int h dx$, where $\left.\der{h}{u^\alpha}\right|_{u^*_*=0}=0$. The local functional~$\oh$ will be denoted by $Q(f)$, so we have obtained a map $Q\colon\cB^{pol}_N\to\Lambda_N$.

\begin{lemma}\label{lemma:pullback}
The pullback of the bracket $\{\cdot,\cdot\}_\eta$ under the map $T_0$ is the bracket $\{\cdot,\cdot\}_{\eta\d_x}$.
\end{lemma}

These lemmas are well-known (see e.g. \cite[Sections 2.2.3 and 2.9.2]{EGH00} and \cite{Ros10}), but we don't know a good reference with proofs. So we decided to give short proofs in Appendix~\ref{appendix}.


\subsection{Extended spaces}\label{subsection:extended spaces}

Introduce a formal indeterminate $\hbar$ of degree $\deg\hbar=-2$. Let $\hcA_{N}:=\cA_N\otimes\mbC[[\hbar]]$ and $\hcA^{[k]}_N\subset\hcA_N$ be the subspace of elements of the total degree $k\ge 0$. The space~$\hcA^{[k]}_N$ consists of elements of the form
$$
f(u;u_x,u_{xx},\ldots;\hbar)=\sum_{i\ge 0}\hbar^i f_i(u;u_x,\ldots),\quad f_i\in\cA_N,\quad\deg_{dif}f_i=2i+k.
$$
The elements of the space $\hcA^{[k]}_N$ will be also called differential polynomials.

Let $\hLambda_N:=\Lambda_N\otimes\mbC[[\hbar]]$ and $\hLambda^{[k]}_N$ be the subspace of elements of the total degree $k$. The space~$\hLambda^{[k]}_N$ consists of integrals of the form
\begin{gather*}
\overline f=\int f(u;u_x,u_{xx},\ldots;\hbar) dx,\quad f\in\hcA^{[k]}_N.
\end{gather*}
They will also be called local functionals. 

Let $K=(K^{\alpha,\beta})_{1\le\alpha,\beta\le N}$ be a matrix of differential operators 
\begin{gather}\label{eq: dif. operator}
K^{\alpha,\beta}=\sum_{i,j\ge 0}f^{\alpha,\beta}_{i,j}\hbar^i\d_x^j,
\end{gather}
where $f^{\alpha,\beta}_{i,j}\in\cA_N$ are homogeneous differential polynomials of the differential degree $\deg_{dif} f^{\alpha,\beta}_{i,j}=2i-j+1$. The bracket $\{\cdot,\cdot\}_K\colon\hLambda^{[k]}_N\times\hLambda^{[l]}_N\to\hLambda^{[k+l+1]}_N$ is defined by
\begin{gather*}
\{\overline g,\overline h\}_K:=\int\sum_{\alpha,\beta}\frac{\delta\overline g}{\delta u^\alpha}K^{\alpha,\beta}\frac{\delta\overline h}{\delta u^\beta}dx.
\end{gather*}
The operator $K$ is called hamiltonian, if the bracket $\{\cdot,\cdot\}_K$ is antisymmetric and satisfies the Jacobi identity. 

A hamiltonian system of partial differential equations is a system of the form
\begin{gather*}
\frac{\d u^\alpha}{\d \tau_i}=\sum_\beta K^{\alpha,\beta}\frac{\delta\oh_i}{\delta u^\beta},\quad 1\le\alpha\le N,\quad i\ge 1,
\end{gather*}
where $K=(K^{\alpha,\beta})$ is a hamiltonian operator and $\oh_i\in\hLambda^{[0]}_N$ are local functionals such that
$\{\oh_i,\oh_j\}_K=0$, for $i,j\ge 1$.

Let $\cB_N^{pol;d}\subset\cB_N^{pol}$ be the subspace that consists of power series of the form
$$
f=\sum_{k\ge 0}\sum_{\substack{1\le\alpha_1,\ldots,\alpha_k\le N\\n_1,\ldots,n_k\ne 0\\n_1+\ldots+n_k=0}}P_{\alpha_1,\ldots,\alpha_k}(n_1,\ldots,n_k)p^{\alpha_1}_{n_1}p^{\alpha_2}_{n_2}\ldots p^{\alpha_k}_{n_k},
$$
where $P_{\alpha_1,\ldots,\alpha_k}(z_1,\ldots,z_k)\in\mbC[z_1,\ldots,z_k]$ are homogeneous polynomials of degree $d$.

Consider the Poisson algebra $\hcB_N:=\cB_N\otimes\mbC[[\hbar]]$. Let $\hcB_N^{pol}\subset\hcB_N$ be the subspace that consists of power series of the form 
$$
f=\sum_{i\ge 0}\hbar^i f_i, \quad f_i\in\cB_N^{pol;2i}.
$$
Repeating the definitions from Section~\ref{subsection:Poisson algebra}, we obtain a map $\hLambda_N^{[0]}\to\hcB_N^{pol}$ that we denote by $\hT_0$. Lemmas~\ref{lemma:surjectivity} and~\ref{lemma:pullback} remain true for this map. Thus, we have a map $\hcB_N^{pol}\to\hLambda_N^{[0]}$ that we denote by $\hQ$.


\section{Geometric preliminaries: cohomological field theories and the double ramification cycles}\label{section:geometric preliminaries}

In this section we recall the definitions of cohomological field theories and the double ramification cycles. We also discuss the polynomiality property of the double ramification cycles that is crucial in our construction of the DR hierarchy.

\subsection{Cohomological field theory}

The notion of cohomological field theory was introduced by M.~Kontsevich and Yu.~Manin in~\cite{KM94}. Let $V$ be a finite dimensional vector space over~$\mbC$, it will be called the phase space. Let us fix a non-degenerate symmetric bilinear form (a scalar product) $(\cdot,\cdot)$ in $V$ and a vector~$\un\in V$, that will be called the unit. Let us denote by $H^*_{even}(\oM_{g,n};\mbC)$ the even part in the cohomology $H^*(\oM_{g,n};\mbC)$. A cohomological field theory is a collection of linear homomorphisms $c_{g,n}: V^{\otimes n}\to H^*_{even}(\oM_{g,n};\mbC)$ defined for all $g$ and $n$ and satisfying the following properties (axioms):

\begin{itemize}

\item $c_{g,n}$ is $S_n$-equivariant, where the group $S_n$ acts on $V^{\otimes n}$ by permutation of the factors, and where its action on $H^*(\oM_{g,n};\mbC)$ is induced by the mappings $\oM_{g,n}\to\oM_{g,n}$ defined by permutation of the marked points.

\item We have $(a,b)=c_{0,3}(\un\otimes a\otimes b)\in H^*(\oM_{0,3};\mbC)$, for all $a,b\in V$. 

\item If $\pi\colon\oM_{g,n+1}\to \oM_{g,n}$ is the forgetful map that forgets the last marked point, then
\begin{gather}\label{eq:forgetful property}
\pi^*c_{g,n}(a_1\otimes\ldots\otimes a_n)=c_{g,n+1}(a_1\otimes\ldots\otimes a_n\otimes\un).
\end{gather}

\item Let $\{e_i\}$ be a basis in $V$ and $\eta_{i,j}:=(e_i,e_j)$. 

\noindent a) If $gl\colon \oM_{g_1,n_1+1}\times\oM_{g_2,n_2+1}\to\oM_{g_1+g_2,n_1+n_2}$ is the gluing map, then
\begin{multline}\label{eq:gluing equation}
gl^*c_{g_1+g_2,n_1+n_2}(a_1\otimes\ldots\otimes a_n)=\\
=\sum_{i,j}c_{g_1,n_1+1}(a_1\otimes\ldots\otimes a_{n_1}\otimes e_i)\cdot c_{g_2,n_2+1}(a_{n_1+1}\otimes\ldots\otimes a_{n_1+n_2}\otimes e_j)\eta^{i,j}.
\end{multline}

\noindent b) If $gl\colon\oM_{g-1,n+2}\to\oM_{g,n}$ is the gluing map, then
$$
gl^*c_{g,n}(a_1\otimes\ldots\otimes a_n)=\sum_{i,j}c_{g-1,n+2}(a_1\otimes\ldots\otimes a_n\otimes e_i\otimes e_j)\eta^{i,j}.
$$

\end{itemize}

The correlators of the cohomological field theory are defined as follows. The class $\psi_i\in H^2(\oM_{g,n};\mbC)$ is defined as the first Chern class of the line bundle over $\oM_{g,n}$ formed by the cotangent lines at the $i$-th marked point (see e.g.~\cite{ACG11} for a rigorous definition). For arbitrary vectors $v_1,v_2,\ldots,v_n\in V$ and any nonnegative integers $d_1,d_2,\ldots,d_n$, let
$$
\<\tau_{d_1}(v_1)\tau_{d_2}(v_2)\ldots\tau_{d_n}(v_n)\>_g:=\int_{\oM_{g,n}} c_{g,n}(v_1\otimes\ldots\otimes v_n)\prod_{i=1}^n\psi_i^{d_i}.
$$

Let us choose a basis~$e_1,\ldots,e_N$ in the phase space $V$. Introduce variables $t^\alpha_d$, where $1\le\alpha\le N$ and $d\ge 0$. Define the potential $F$ of the cohomological field theory by 
\begin{align*}
&F:=\sum_{g\ge 0}\hbar^g F_g,\text{ where}\\
&F_g:=\sum_{\substack{n\ge 0\\2g-2+n>0}}\frac{1}{n!}\sum_{\substack{1\le\alpha_1,\ldots,\alpha_n\le N\\d_1,\ldots,d_n\ge 0}}\<\tau_{d_1}(e_{\alpha_1})\tau_{d_2}(e_{\alpha_2})\ldots\tau_{d_n}(e_{\alpha_n})\>_g\prod_{i=1}^n t_{d_i}^{\alpha_i}.
\end{align*}

Assume that $e_1$ is the unit $\un$. Recall the string equation:
\begin{gather}\label{eq:string}
\frac{\d F}{\d t^1_0}=\sum_{\substack{1\le \alpha\le N\\p\ge 0}}t^\alpha_{p+1}\frac{\d F}{\d t^\alpha_p}+\frac{1}{2}\sum_{\alpha,\beta}\eta_{\alpha,\beta}t^\alpha_0 t^\beta_0+\hbar\left<\tau_0(e_1)\right>_1.
\end{gather}
It is a simple consequence of~\eqref{eq:forgetful property} and the formula for the pull-back of $\psi_i$ under the forgetful map $\pi\colon\oM_{g,n+1}\to\oM_{g,n}$ (see e.g.~\cite{ACG11}).

\subsubsection{Examples}

Let us give two simple examples of a cohomological field theory. Let $V$ be of dimension $1$ with a basis vector $e$. Define a scalar product by $(e,e)=1$ and let $e$ be a unit~$\un$. 

The trivial cohomological field theory is given by the formula 
$$
c^{triv}_{g,n}(e\otimes\ldots\otimes e):=1\in H^*(\oM_{g,n};\mbC).
$$

We will also consider the cohomological field theory formed by the Hodge classes:
$$
c^{Hodge}_{g,n}(e\otimes\ldots\otimes e):=1+\eps\lambda_1+\ldots+\eps^g\lambda_g\in H^*(\oM_{g,n};\mbC),
$$
where $\lambda_j\in H^{2j}(\oM_{g,n};\mbC)$ is the $j$-th Chern class of the rank $g$ Hodge vector bundle over $\oM_{g,n}$ whose fibers are the spaces of holomorphic one-forms. The fact, that the classes $1+\eps\lambda_1+\ldots+\eps^g\lambda_g$ form a cohomological field theory, was first noticed in~\cite{Mum83}.

For these two cohomological field theories we will explicitly compute the DR hierarchy and compare it with the Dubrovin-Zhang hierarchy.


\subsection{Double ramification cycles}

In this section we recall the definition of the double ramification cycles and formulate the polynomiality property that is crucial in our construction of the DR hierarchy.

\subsubsection{Definition}

Let $a_1,\ldots,a_n$ be a list of integers satisfying $\sum a_i=0$ and assume that not all of them are equal to zero. To such a list we assign a space of ``rubber'' stable maps to~$\mbP^1$ relative to $0$ and~$\infty$ in the following way.

Denote by $n_+$ the number of positive integers among the $a_i$'s. They form a partition $\nu = (\nu_1,\dots,\nu_{n_+})$. Similarly, denote by $n_-$ the number of negative integers among the $a_i$'s. After a change of sign they form another partition $\mu =(\mu_1,\ldots,\mu_{n_-})$. Both $\mu$ and $\nu$ are partitions of the same integer $d:=\frac{1}{2} \sum_{i=1}^n |a_i|$. Finally, let $n_0$ be the number of vanishing $a_i$'s. 

To the list $a_1,\ldots,a_n$ we assign the space 
$$
\oM_{g;a_1,\dots,a_n}:=\oM_{g,n_0;\mu,\nu}(\mbP^1,0,\infty)
$$ 
of degree~$d$ ``rubber'' stable maps to $\mbP^1$ relative to $0$ and $\infty$ with the ramification profiles $\mu$ and $\nu$ respectively (see e.g.~\cite{GV05}). Here ``rubber'' means that we factor the space by the $\mbC^*$-action in the target~$\mbP^1$. We consider the pre-images of $0$ and $\infty$ as marked points and there are $n_0$ more additional marked points.

Thus, in the source curve there are $n$ numbered marked points with labels $a_1,\dots,a_n$. The relative stable map sends the points with positive labels to $\infty$, those with negative labels to~$0$, while those with zero labels do not have a fixed image.

We have the forgetful map $st\colon\oM_{g;a_1,\ldots,a_n}\to\oM_{g,n}$. The space space $\oM_{g;a_1,\ldots,a_n}$ has the virtual fundamental class $[\oM_{g;a_1,\ldots,a_n}]^\virt\in H_{2(2g-3+n)}(\oM_{g;a_1,\ldots,a_n};\mbC)$.

\begin{definition}
The push-forward $st_*[\oM_{g;a_1,\ldots,a_n}]^\virt\in H_{2(2g-3+n)}(\oM_{g,n};\mbC)$ of the virtual fundamental class under the forgetful map~$st$ is called the double ramification cycle or the DR-cycle and is denoted by $\DR_g(a_1,\ldots,a_n)$. 
\end{definition}
These classes were introduced by T. Graber and R. Vakil in~\cite{GV05}. It is known (see~\cite{FP05}) that the Poincar\'e dual cohomology class of $\DR_g(a_1, \dots, a_n)$ lies in the tautological ring of~$\oM_{g,n}$. 

Let $\pi\colon\oM_{g,n+1}\to\oM_{g,n}$ be the forgetful map that forgets the last marked point. As an immediate consequence of the definition, we get the following property:
$$
\pi^* DR_g(a_1,\ldots,a_n)=DR_g(a_1,\ldots,a_n,0).
$$

In genus $0$ the double ramification cycle coincides with the fundamental class of the moduli space of curves (see e.g.~\cite{GJV11}):
$$
DR_0(a_1,\ldots,a_n)=[\oM_{0,n}].
$$

\subsubsection{Polynomiality}

There is another version of the double ramification cycles defined using the universal Jacobian over the moduli space $\M_{g,n}^{ct}$ of stable curves of compact type. This is a class in the Borel-Moore homology $H^{BM}_{2(2g-3+n)}(\M_{g,n}^{ct};\mbC)$ and we will denote it by $DR_g^{Jac}(a_1,\ldots,a_n)$. R. Hain~\cite{Hain13} obtained an explicit formula for this class that we are going to recall here (see also~\cite{GZ12} for a different proof). 

Let $\delta^J_h$ be the class of the divisor whose generic point is a reducible curve consisting of a smooth component of genus $h$ containing the marked points indexed by $J$ and a smooth component of genus $g-h$ with the remaining points, joined at a node. Denote by $\psi^\dagger_i$ the $\psi$-class that is pulled back from $\M^{ct}_{g,1}$.

Now we can state Hain's result:
\begin{gather}
DR_g^{Jac}(a_1,\ldots,a_n)=\frac{1}{g!}\left(\sum_{j=1}^n\frac{a_j^2\psi_j^\dagger}{2}-\sum_{\substack{J\subset\{1,\ldots,n\}\\|J|\ge 2}}\left(\sum_{i,j\in J,i<j}a_ia_j\right)\delta_0^J-\frac{1}{4}\sum_{J\subset\{1,\ldots,n\}}\sum_{h=1}^{g-1}a_J^2\delta_h^J\right)^g,\label{eq:Hain}
\end{gather}
where $a_J:=\sum_{j\in J}a_j$.

Happily, in \cite{MW13} it is proved that$\left.DR_g(a_1,\ldots,a_n)\right|_{\M_{g,n}^{ct}}=DR_g^{Jac}(a_1,\ldots,a_n)$. Since the class $\lambda_g$ vanishes on $\oM_{g,n}\setminus\M_{g,n}^{ct}$, we get the following lemma that is crucial in our construction.
\begin{lemma}\label{lemma:polynomialty}
For any cohomology class $\alpha\in H^*(\oM_{g,n};\mbC)$ the integral $\int_{DR_g(a_1,\ldots,a_n)}\lambda_g\alpha$ is a homogeneous polynomial in $a_1,\ldots,a_n$ of degree $2g$.
\end{lemma}


\section{Main construction}\label{section:construction}

In this section we give the construction of the DR hierarchy. We discuss its main properties and explicitly compute it in two examples. 

Consider an arbitrary cohomological field theory $c_{g,n}\colon V^{\otimes n}\to H^*_{even}(\oM_{g,n};\mbC)$. Let $V$ be the phase space and $\dim V=N$. Let $e_1,e_2,\ldots,e_N$ be a basis of $V$ such that $e_1$ is the unit. Denote by $\eta=(\eta_{\alpha,\beta})$, $\eta_{\alpha,\beta}:=(e_\alpha,e_\beta)$, the matrix of the scalar product. 

\subsection{Hamiltonians}\label{subsection:Hamiltonians}

For any $1\le\alpha\le N$ and $d\ge 0$, define the power series $g_{\alpha,d}\in\hcB_N$ by
\begin{gather}\label{eq:definition of g}
g_{\alpha,d}=\sum_{g\ge 0}\sum_{n\ge 2}\frac{(-\hbar)^g}{n!}\sum_{\substack{a_1,\ldots,a_n\ne 0\\a_1+\ldots+a_n=0\\1\le\alpha_1,\ldots,\alpha_n\le N}}\left(\int_{DR_g(0,a_1,\ldots,a_n)}\lambda_g\psi_1^d c_{g,n+1}(e_\alpha\otimes e_{\alpha_1}\otimes\ldots\otimes e_{\alpha_n})\right)\prod_{i=1}^n p_{a_i}^{\alpha_i}.
\end{gather}
From Lemma~\ref{lemma:polynomialty} it follows that $g_{\alpha,d}\in\hcB_N^{pol}$. Define the local functional $\overline g_{\alpha,d}\in\hLambda^{[0]}_N$ by
$$
\overline g_{\alpha,d}:=\hQ(g_{\alpha,d}).
$$

The following theorem is the main result of the paper.
\begin{theorem}\label{theorem:main}
For any $1\le\alpha_1,\alpha_2\le N$ and $d_1,d_2\ge 0$, we have
$$
\{\overline g_{\alpha_1,d_1},\overline g_{\alpha_2,d_2}\}_{\eta\d_x}=0.
$$
\end{theorem}
We prove the theorem in Section~\ref{section:commutativity}. The hamiltonian hierarchy of PDEs corresponding to the operator $\eta\d_x$ and the local functionals $\og_{\alpha,d}$ is called the DR hierarchy:
\begin{gather}\label{eq: DR hierarchy}
\der{u^\beta}{t^\alpha_d}=\sum_\mu\eta^{\beta,\mu}\d_x\frac{\delta\og_{\alpha,d}}{\delta u^\mu}.
\end{gather}

\begin{remark}
Our construction of the DR hierarchy is very similar to the construction of the quantum hierarchy in Symplectic Field Theory~\cite{EGH00}. The crucial difference is the presence of the class $\lambda_g$ in the integrand on the right-hand side of~\eqref{eq:definition of g}. This allows us to construct a classical hierarchy instead of the quantum hierarchy from SFT.
\end{remark}


\subsection{Main properties}

In this section we list several properties of the DR hierarchy.

\subsubsection{Hamiltonian $\og_{1,0}$}\label{subsection:dx flow}

\begin{lemma}\label{lemma:dx-flow}
We have $\og_{1,0}=\frac{1}{2}\int\left(\sum_{\alpha,\beta}\eta_{\alpha,\beta}u^\alpha u^\beta\right)dx$.
\end{lemma}
Therefore, the first equation of the DR hierarchy is
$$
\frac{\d u^\alpha}{\d t^1_0}=u^\alpha_x.
$$
\begin{proof}[Proof of Lemma~\ref{lemma:dx-flow}]
Suppose $2g-2+n>0$ and let $\pi\colon\oM_{g,n+1}\to\oM_{g,n}$ be the forgetful map that forgets the first marked point. We have
\begin{align*}
&DR_g(0,a_1,\ldots,a_n)=\pi^* DR_g(a_1,\ldots,a_n),\\
&c_{g,n+1}(e_1\otimes e_{\alpha_1}\otimes\ldots\otimes e_{\alpha_n})=\pi^* c_{g,n}(e_{\alpha_1}\otimes\ldots\otimes e_{\alpha_n}).
\end{align*} 
Hence, $\int_{DR_g(0,a_1,\ldots,a_n)}\lambda_g c_{g,n+1}(e_1\otimes e_{\alpha_1}\otimes\ldots\otimes e_{\alpha_n})=0$. 

In the case $g=0$ and $n=2$ we have
\begin{align*}
&\int_{DR_0(0,a,-a)}c_{0,3}(e_1\otimes e_\alpha\otimes e_\beta)=\eta_{\alpha,\beta}, \text{ therefore,}\\
&g_{1,0}=\frac{1}{2}\sum_{\alpha,\beta}\sum_{a\ne 0}\eta_{\alpha,\beta}p^\alpha_a p^\beta_{-a}=\frac{1}{2}\hT_0\left(\int\sum_{\alpha,\beta}\eta_{\alpha,\beta}u^\alpha u^\beta dx\right).
\end{align*}
This completes the proof of the lemma.
\end{proof}

\subsubsection{Genus $0$ part}\label{subsection:genus 0}

Here we compute the genus $0$ part of the DR hierarchy and compare it with the genus $0$ part of the Dubrovin-Zhang hierarchy. Let $\og^{[0]}_{\alpha,d}:=\left.\og_{\alpha,d}\right|_{\hbar=0}$ and  
$$
\Omega^{[0]}_{\alpha,p;\beta,q}(u):=\left.\frac{\d^2 F_0}{\d t^\alpha_p\d t^\beta_q}\right|_{\substack{t^*_{\ge 1}=0\\t_0^\mu=u^\mu}}.
$$

\begin{lemma}
We have $\og^{[0]}_{\alpha,d}=\int\Omega^{[0]}_{\alpha,d+1;1,0} dx$.
\end{lemma}
\begin{proof}
We have
\begin{multline*}
\og^{[0]}_{\alpha,d}=\sum_{n\ge 2}\frac{1}{n!}\sum_{\substack{a_1,\ldots,a_n\ne 0\\a_1+\ldots+a_n=0\\1\le\alpha_1,\ldots,\alpha_n\le N}}\left(\int_{DR_0(0,a_1,\ldots,a_n)}\psi_1^d c_{0,n+1}(e_\alpha\otimes e_{\alpha_1}\otimes\ldots\otimes e_{\alpha_n})\right)\prod_{i=1}^n
p^{\alpha_i}_{a_i}=\\
=\sum_{n\ge 2}\frac{1}{n!}\sum_{\substack{a_1,\ldots,a_n\ne 0\\a_1+\ldots+a_n=0\\1\le\alpha_1,\ldots,\alpha_n\le N}}\<\tau_d(e_\alpha)\tau_0(e_{\alpha_1})\ldots\tau_0(e_{\alpha_n})\>_0\prod_{i=1}^np^{\alpha_i}_{a_i}=\hT_0\left(\int\Omega^{[0]}_{\alpha,d}(u)dx\right),
\end{multline*}
where $\Omega^{[0]}_{\alpha,d}(u):=\left.\frac{\d F_0}{\d t^\alpha_d}\right|_{\substack{t^*_{\ge 1}=0\\t_0^\mu=u^\mu}}$. From the string equation~\eqref{eq:string} it follows that $\Omega^{[0]}_{\alpha,d}=\Omega^{[0]}_{\alpha,d+1;1,0}$. The lemma is proved.
\end{proof}

Suppose our cohomological field theory is semisimple. Let $\oh_{\alpha,d}^{DZ}$ and $K^{DZ}=(K^{DZ;\alpha,\beta})$ be the local functionals and the hamiltonian operator of the corresponding Dubrovin-Zhang hierarchy. We have $\oh_{\alpha,d}^{DZ}=\int\left(\Omega^{[0]}_{\alpha,d+1;1,0}+O(\hbar)\right)dx$ and $K^{DZ}=\eta\d_x+O(\hbar)$ (see \cite{BPS12a}). We see that the genus $0$ parts of the Dubrovin-Zhang and the DR hierarchies coincide. This agrees with our conjecture from the introduction.

\begin{remark}
In genus $0$ the DR hierarchy coincides with the quantum hierarchy from Symplectic Field Theory. The fact, that the genus $0$ part of the quantum hierarchy coincides with the genus $0$ part of the Dubrovin-Zhang hierarchy, was first noticed in F.~Bourgeois's thesis and then mentioned in~\cite{EGH00}.
\end{remark}

\subsubsection{String equation for the DR hierarchy}

\begin{lemma}\label{lemma:string equation}
We have
\begin{gather*}
\der{\og_{\alpha,d}}{u^1}=
\begin{cases}
\og_{\alpha,d-1},&\text{ if $d\ge 1$},\\
\int\sum_\mu\eta_{\alpha,\mu}u^\mu dx,&\text{ if $d=0$}.
\end{cases}
\end{gather*}
\end{lemma}
This equation is analogous to the string equation for the quantum hierarchy in Symplectic Field Theory, that was proved in~\cite{FR11}.
\begin{proof}
We are going to use the material from Section~\ref{appendix:auxiliary lemmas}. The spaces $\cB'_N$, ${\cB'}_N^{pol}$ and the map $Z\colon\cB_N^{pol}\to{\cB'}_N^{pol}$ have obvious analogs $\widehat\cB'_N$, $\widehat{\cB'}_N^{pol}$ and $\hZ\colon\hcB_N^{pol}\to\widehat{\cB'}_N^{pol}$. An analog of Lemma~\ref{lemma:Z map} is clearly true for any local functional $\oh\in\hLambda^{[0]}_N$. So, we have
$$
\hT_0\left(\frac{\d\og_{\alpha,d}}{\d u^1}\right)=\left.\frac{\d\hZ(\hT_0(\og_{\alpha,d}))}{\d p^1_0}\right|_{p^*_0=0}=\left.\frac{\d\hZ(g_{\alpha,d})}{\d p^1_0}\right|_{p^*_0=0}.
$$ 
We obviously have
$$
\hZ(g_{\alpha,d})=\sum_{g\ge 0}\sum_{n\ge 2}\frac{(-\hbar)^g}{n!}\sum_{\substack{a_1,\ldots,a_n\in\mbZ\\a_1+\ldots+a_n=0\\1\le\alpha_1,\ldots,\alpha_n\le N}}\left(\int_{DR_g(0,a_1,\ldots,a_n)}\lambda_g\psi_1^d c_{g,n+1}(e_\alpha\otimes e_{\alpha_1}\otimes\ldots\otimes e_{\alpha_n})\right)\prod_{i=1}^n p_{a_i}^{\alpha_i}.
$$
Therefore,
\begin{align}
\left.\frac{\d\hZ(g_{\alpha,d})}{\d p^1_0}\right|_{p^*_0=0}=&\sum_{g\ge 0}\sum_{n\ge 2}\frac{(-\hbar)^g}{n!}\times\notag\\
&\times\sum_{\substack{a_1,\ldots,a_n\ne 0\\a_1+\ldots+a_n=0\\1\le\alpha_1,\ldots,\alpha_n\le N}}\left(\int_{DR_g(0,0,a_1,\ldots,a_n)}\lambda_g\psi_1^d c_{g,n+2}(e_\alpha\otimes e_1\otimes e_{\alpha_1}\otimes\ldots\otimes e_{\alpha_n})\right)\prod_{i=1}^n p_{a_i}^{\alpha_i}.\label{integral}
\end{align}
Let $\pi_2\colon\oM_{g,n+2}\to\oM_{g,n+1}$ be the forgetful map that forgets the second marked point. Since  $(\pi_2)_*\psi_1^d=\psi_1^{d-1}$, if $d\ge 1$, and $(\pi_2)_*\psi_1^0=0$, we get
\begin{multline*}
\int_{DR_g(0,0,a_1,\ldots,a_n)}\lambda_g\psi_1^d c_{g,n+2}(e_\alpha\otimes e_1\otimes e_{\alpha_1}\otimes\ldots\otimes e_{\alpha_n})=\\=\begin{cases}
\int_{DR_g(0,a_1,\ldots,a_n)}\lambda_g\psi_1^{d-1} c_{g,n+1}(e_\alpha\otimes e_1\otimes\ldots\otimes e_{\alpha_n}),&\text{if $d\ge 1$},\\
0,&\text{if $d=0$}.
\end{cases}
\end{multline*}
We obtain
$$
\hT_0\left(\frac{\d\og_{\alpha,d}}{\d u^1}\right)=
\begin{cases}
\hT_0\left(\og_{\alpha,d-1}\right),&\text{ if $d\ge 0$},\\
0,&\text{ if $d=0$}.
\end{cases}
$$

Suppose $d\ge 1$. From Lemma~\ref{lemma:surjectivity} it follows that $\frac{\d\og_{\alpha,d}}{\d u^1}-\og_{\alpha,d-1}$ is a linear combination of the local {functionals}~$\int u^\beta dx$. Since $\int_{DR_0(0,a,-a)}\psi_1^d c_{0,3}(e_\alpha\otimes e_{\alpha_1}\otimes e_{\alpha_2})=0$, we get $\frac{\d\og_{\alpha,d}}{\d u^1}-\og_{\alpha,d-1}=0$.   

Suppose $d=0$. By Lemma~\ref{lemma:surjectivity}, $\frac{\d\og_{\alpha,0}}{\d u^1}$ is a linear combination of the local functionals $\int u^\beta dx$. Denote $c_{0,3}(e_\alpha\otimes e_{\beta}\otimes e_{\gamma})$ by $c_{\alpha,\beta,\gamma}$. We have
$$
\frac{1}{2}\sum_{\substack{a\ne 0\\1\le\alpha_1,\alpha_2\le N}}\left(\int_{DR_0(0,a,-a)}c_{0,3}(e_\alpha\otimes e_{\alpha_1}\otimes e_{\alpha_2})\right)p_{a}^{\alpha_1}p_{-a}^{\alpha_2}=T_0\left(\frac{1}{2}\int\sum_{\mu,\nu}c_{\alpha,\mu,\nu}u^\mu u^\nu dx\right).
$$
Since $c_{1,\alpha,\beta}=\eta_{\alpha,\beta}$, we get $\frac{\d\og_{\alpha,0}}{\d u^1}=\int\sum_{\mu}\eta_{\alpha,\mu}u^\mu dx$. The lemma is proved.
\end{proof}

\subsubsection{String solution of the DR hierarchy}

In this section we show that the DR hierarchy has a special solution that satisfies the string equation.
  
Let $u^{str}(x;t^*_*;\hbar)$ be the solution of the DR hierarchy~\eqref{eq: DR hierarchy} specified by the initial condition $\left.(u^{str})^\alpha\right|_{t^*_*=0}=\delta_{\alpha,1}x$. 
\begin{lemma}\label{lemma:string solution}
We have 
$$
\frac{\d (u^{str})^\alpha}{\d t^1_0}=\sum_{\substack{1\le\mu\le N\\d\ge 0}}t^\mu_{d+1}\frac{\d (u^{str})^\alpha}{\d t^\mu_d}+\delta_{\alpha,1}.
$$
\end{lemma}
\begin{proof}
Define an operator $L$ by $L:=\frac{\d}{\d t^1_0}-\sum_{\mu,d}t^\mu_{d+1}\frac{\d}{\d t^\mu_d}$. Let $f^\alpha_{\beta,q}:=\sum_\mu\eta^{\alpha,\mu}\d_x\frac{\delta\og_{\beta,q}}{\delta u^\mu}$. We have
$$
\der{}{t^\beta_q}\left(L(u^{str})^\alpha\right)=L\der{(u^{str})^\alpha}{t^\beta_q}-\der{(u^{str})^\alpha}{t^\beta_{q-1}}=L f^\alpha_{\beta,q}-f^\alpha_{\beta,q-1}=\sum_{\gamma,n}\der{f^\alpha_{\beta,q}}{u^\gamma_n}\d_x^n (L u^\gamma)-f^\alpha_{\beta,q-1}.
$$
Here we, by definition, put $\der{}{t^\beta_{-1}}:=0$ and $f^\alpha_{\beta,-1}:=0$. We get the system
$$
\der{}{t^\beta_q}\left(L(u^{str})^\alpha\right)=\sum_{\gamma,n}\der{f^\alpha_{\beta,q}}{u^\gamma_n}\d_x^n(L u^\gamma)-f^\alpha_{\beta,q-1},\quad 1\le\alpha,\beta\le N,\quad q\ge 0.
$$
This system, together with the initial condition $\left.L(u^{str})^\alpha\right|_{t^*_*=0}=\delta_{\alpha,1}$, uniquely determines~$L(u^{str})^\alpha$. By Lemma~\ref{lemma:string equation}, $\frac{\d f^\alpha_{\beta,q}}{\d u^1}=f^\alpha_{\beta,q-1}$ and, therefore, $L(u^{str})^\alpha=\delta_{\alpha,1}$ satisfies the system. This concludes the proof of the lemma.
\end{proof}

\subsubsection{Local functional $\og$}

Introduce the local functional $\og\in\hLambda^{[0]}_N$ by $\og:=\hQ(g)$, where
$$
g:=\sum_{g\ge 0}\sum_{\substack{n\ge 2\\2g-2+n>0}}\frac{(-\hbar)^g}{n!}\sum_{\substack{a_1,\ldots,a_n\ne 0\\a_1+\ldots+a_n=0\\1\le\alpha_1,\ldots,\alpha_n\le N}}\left(\int_{DR_g(a_1,\ldots,a_n)}\lambda_g c_{g,n}(e_{\alpha_1}\otimes\ldots\otimes e_{\alpha_n})\right)\prod_{i=1}^n p_{a_i}^{\alpha_i}.
$$
The local functional $\og$ can be easily related to $\og_{1,1}$ as follows. Define a differential operator~$O$ on the space of differential polynomials $\hcA^{[0]}_N$ by $O:=2\hbar\frac{\d}{\d\hbar}-2+\sum_{\gamma,n}u^\gamma_n\frac{\d}{\d u^\gamma_n}$. Since the operator $\sum_{\gamma,n}u^\gamma_n\frac{\d}{\d u^\gamma_n}$ commutes with $\d_x$, the operator~$O$ is well defined on the space of local functionals~$\hLambda_N^{[0]}$. We claim that 
\begin{gather}\label{eq:dilaton equation}
\og_{1,1}=O\og.
\end{gather}
This equation can be regarded as analogous to the dilaton equation for the quantum hierarchy in Symplectic Field Theory, that was proved in~\cite{FR11}. Formula~\eqref{eq:dilaton equation} easily follows from the equation
\begin{multline*}
\int_{DR_g(0,a_1,\ldots,a_n)}\lambda_g \psi_1 c_{g,n+1}(e_1\otimes e_{\alpha_1}\otimes\ldots\otimes e_{\alpha_n})=\\
=\begin{cases}
(2g-2+n)\int_{DR_g(a_1,\ldots,a_n)}\lambda_g c_{g,n}(e_{\alpha_1}\otimes\ldots\otimes e_{\alpha_n}),&\text{if $2g-2+n>0$},\\
0,&\text{otherwise}.
\end{cases}
\end{multline*}

\begin{lemma}\label{lemma:derivative}
We have $\og_{\alpha,0}=\der{\og}{u^\alpha}$.
\end{lemma}
\begin{proof}
The proof goes in the same way as the proof of Lemma~\ref{lemma:string equation}. We have
\begin{align*}
&\hT_0\left(\frac{\d\og}{\d u^\alpha}\right)\stackrel{\text{by Lemma~\ref{lemma:Z map}}}{=}\left.\frac{\d\hZ(\hT_0(\og))}{\d p^\alpha_0}\right|_{p^*_0=0}=\left.\frac{\d\hZ(g)}{\d p^\alpha_0}\right|_{p^*_0=0}=\\
=&\frac{\d}{\d p^\alpha_0}\left.\left(\sum_{g\ge 0}\sum_{\substack{n\ge 2\\2g-2+n>0}}\frac{(-\hbar)^g}{n!}\sum_{\substack{a_1,\ldots,a_n\in\mbZ\\a_1+\ldots+a_n=0\\1\le\alpha_1,\ldots,\alpha_n\le N}}\left(\int_{DR_g(a_1,\ldots,a_n)}\lambda_g c_{g,n}(e_{\alpha_1}\otimes\ldots\otimes e_{\alpha_n})\right)\prod_{i=1}^n p_{a_i}^{\alpha_i}\right)\right|_{p^*_0=0}=\\
=&\hT_0(\og_{\alpha,0}).
\end{align*}
By Lemma~\ref{lemma:surjectivity}, $\frac{\d\og}{\d u^\alpha}-\og_{\alpha,0}$ is a linear combination of the local functionals $\int u^\beta dx$. It is easy to see that 
$$
\left.\og\right|_{\hbar=0}=\int\left(\sum_{\alpha_1,\alpha_2,\alpha_3}\frac{1}{6}c_{\alpha_1,\alpha_2,\alpha_3}u^{\alpha_1}u^{\alpha_2}u^{\alpha_3}+O(u^4))\right)dx.
$$
We conclude that $\frac{\d\og}{\d u^\alpha}-\og_{\alpha,0}=0$. The lemma is proved.
\end{proof}


\subsection{Examples}

Here we explicitly compute the DR hierarchy in two examples. In both of them the phase space of a cohomological field theory will be one-dimensional. So, we will omit the first index in the Hamiltonians $\og_{\alpha,d}$ and denote them by $\og_d$.

\subsubsection{Trivial CohFT}

Consider the trivial cohomological field theory: $c^{triv}_{g,n}=1$. From Section~\ref{subsection:genus 0} we know that $\og_i=\int\left(\Omega_{1,i+1;1,0}^{[0]}+O(\hbar)\right)dx$. We have $\Omega_{1,i+1;1,0}^{[0]}=\frac{u^{i+2}}{(i+2)!}$ (see e.g.~\cite{BPS12a}), therefore,
$$
\og_i=\int\left(\frac{u^{i+2}}{(i+2)!}+O(\hbar)\right)dx.
$$

Let us compute the Hamiltonian $\og_1$. We have to compute the integrals
\begin{gather}\label{eq:integral for kdv}
\int_{DR_g(0,a_1,\ldots,a_n)}\lambda_g\psi_1.
\end{gather}
We have $DR_g(0,a_1,\ldots,a_n)\in H_{2(2g-2+n)}(\oM_{g,n+1};\mbC)$. Therefore, the integral \eqref{eq:integral for kdv} can be nonzero, only if $g=0$ and $n=3$, or $g=1$ and $n=2$. We already know the genus~$0$ integral. Let us compute the genus $1$ integral:
$$
\int_{DR_1(0,a,-a)}\lambda_1\psi_1=2\int_{DR_1(a,-a)}\lambda_1\stackrel{\text{by \eqref{eq:Hain}}}{=}2a^2\int_{\delta_0^{\{1,2\}}}\lambda_1=\frac{a^2}{12}.
$$
Thus, $\og_1=\int\left(\frac{u^3}{6}+\frac{\hbar}{24}u u_{xx}\right)dx$. This is the first Hamiltonian of the KdV hierarchy. From \cite[Lemma~2.4]{Bur13} it follows that the higher Hamiltonians $\og_i$, $i\ge 2$, coincide with the higher Hamiltonians of the KdV hierarchy. We conclude that the DR hierarchy for the trivial cohomological field theory coincides with the KdV hierarchy. 

The Dubrovin-Zhang hierarchy corresponding to the trivial cohomological field theory also coincides with the KdV hierarchy (see \cite{DZ05}). This agrees with the conjecture that we suggested in the introduction.


\subsubsection{Hodge classes}

Consider the cohomological field theory formed by the Hodge classes: $c^{Hodge}_{g,n}=1+\eps\lambda_1+\ldots+\eps^g\lambda_g$. Again we have $\og_i=\int\left(\frac{u^{i+2}}{(i+2)!}+O(\hbar)\right)dx$. Let us compute the Hamiltonian $\og_1$. We have to compute the integrals
\begin{gather}\label{eq:integral for hodge}
\int_{DR_g(0,a_1,\ldots,a_n)}\lambda_g(1+\eps\lambda_1+\ldots+\eps^g\lambda_g)\psi_1.
\end{gather}
The coefficient of $\eps^j$ can be nonzero, only if $2g-2+n=g+j+1$, or, equivalently, $g-j=3-n$. Note that $\lambda_g^2=0$, for $g\ge 1$. Thus, the coefficient of $\eps^j$ in the integral~\eqref{eq:integral for hodge} can be nonzero, only if $g=j=0$ and $n=3$, or $g\ge 1$, $j=g-1$ and $n=2$.  For $g\ge 1$, we have
$$
\int_{DR_g(0,a,-a)}\lambda_g\lambda_{g-1}\psi_1=2g\int_{DR_g(a,-a)}\lambda_g\lambda_{g-1}=a^{2g}\frac{|B_{2g}|}{(2g)!}.
$$
Here $B_{2g}$ are Bernoulli numbers: $B_2=\frac{1}{6},B_4=-\frac{1}{30},\ldots$; and the computation of the last integral can be found, for example, in~\cite{CMW12}. We get $\og_1=\int\left(\frac{u^3}{6}+\sum_{g\ge 1}\hbar^g\eps^{g-1}\frac{|B_{2g}|}{2(2g)!}u u_{2g}\right)dx$. We conclude that our DR hierarchy coincides with the deformed KdV hierachy (see~\cite{Bur13} and Lemma~2.4 there). 

Consider the Dubrovin-Zhang hierarchy corresponding to our cohomological field theory. In~\cite{Bur13} it is proved that the Miura transformation
$$
u\mapsto \widetilde u=u+\sum_{g\ge 1}\frac{(-1)^g}{2^{2g}(2g+1)!}\hbar^{g}\eps^g u_{2g}
$$
transforms it to the deformed KdV hierarchy. This again agrees with our conjecture from the introduction.


\section{Commutativity of the Hamiltonians}\label{section:commutativity}

In this section we prove Theorem~\ref{theorem:main}. In Section~\ref{subsection:Losev-Manin} we introduce the Losev-Manin moduli space $\LM_{r+n_0}$ and a map $q\colon \oM_{g;a_1,\ldots,a_n}\to\LM_{r+n_0}/S_r$. In Section~\ref{subsection:pullbacks} we compute the pullbacks of certain divisors in $\LM_{r+n_0}/S_r$. In Section~\ref{subsection:proof of the main theorem} we show that an equation of certain divisors in $\LM_{r+n_0}$ implies Theorem~\ref{theorem:main}. 

We recommend the reader the paper~\cite{BSSZ12} for a more detailed discussion of the geometric constructions that we use here.

\subsection{Losev-Manin moduli space}\label{subsection:Losev-Manin}

The Losev-Manin moduli space $\LM_r$ is a compactification of $\M_{0,r+2}$. It is the moduli space of chains of spheres with two special ``white'' marked points $0$ and $\infty$ at the extremities of the chain and $r$ more ``black'' marked points in the other spheres. The black points are allowed to coincide with each other and there should be at least one black point per sphere. For more details see~\cite{LM00}.

We have two forgetful maps from the DR-space $\oM_{g;a_1,\ldots,a_n}$:
$$
\LM_{r+n_0}/S_r\stackrel{q}{\longleftarrow}\oM_{g;a_1,\ldots,a_n}\stackrel{st}{\longrightarrow}\oM_{g,n},
$$
where $n_0$ is the number of indices $i$ such that $a_i=0$ and $r:=2g-2+n$ is the number of branch points.

The map~$st$ assigns to a relative stable map its stabilized source curve. This is the map that we used to define the DR-cycle $\DR_g(a_1,\ldots,a_n)$.

Let us describe the map $q$. It is very similar to the branch morphism (see e.g.~\cite{GJV11}). By~$S_r$ we denote the symmetric group. It acts on $\LM_{r+n_0}$ by permutation of the first $r$ black marked points. Suppose $f\colon C\to T$ is a relative stable map. The map $q$ assigns to $f$ its target rational curve $T$. Since $f\colon C\to T$ is a relative stable map, we already have the points $0$ and~$\infty$ in $T$. Therefore, it remains to choose black marked points in $T$. Suppose that over each irreducible component of~$T$ the map $f$ is a ramified covering. Then the black marked points in~$T$ are the~$r$ branch points and the images of the marked points with zero labels in the source curve $C$. The fact, that the branch points are not numbered, is the reason that we have to take the quotient of the Losev-Manin space by the action of the symmetric group.  


\subsection{Pullbacks of divisors}\label{subsection:pullbacks}

Consider the space $\oM_{g;0,0,a_1,\ldots,a_n}$, where $a_1,\ldots,a_n\ne 0$. We denote by $p_1$ and $p_2$ the first two ``free'' marked points in a curve from $\oM_{g;0,0,a_1,\ldots,a_n}$. The images of these points in a curve from $\LM_{r+2}/S_r$ will be denoted by $x_1$ and $x_2$ correspondingly.

Denote by $D_{(0,x_1|x_2,\infty)}\subset\LM_{r+2}$ the divisor of two-component curves, where the pairs of points~$0,x_1$ and $x_2,\infty$ lie in different components. Let $D^{Sym}_{(0,x_1|x_2,\infty)}$ be its symmetrization in~$\LM_{r+2}/S_r$. Let us compute the class $st_*q^* D^{Sym}_{(0,x_1|x_2,\infty)}$. In order to do it we have to introduce a bit more notations.

Let $b_1,\ldots,b_m$ be some integers such that $b_1+\ldots+b_m=0$ and not of all them are equal to zero. Let us divide the list $b_1,\ldots,b_m$ into two non-empty disjoint parts, $I\sqcup J = \{1,\ldots,m\}$, in such a way that $\sum_{i\in I}b_i<0$ or, equivalently, $\sum_{j\in J} b_j>0$. Then we choose a list of positive integers~$k_1,\ldots, k_p$ in such a way that 
$$
\sum_{i \in I}b_i+\sum_{i=1}^p k_i=\sum_{j\in J}b_j-\sum_{i=1}^p k_i = 0.
$$
Now we take two DR-cycles $\DR_{g_1}(b_I,k_1,\ldots,k_p)$ and $\DR_{g_2}(-k_1,\ldots,-k_p,b_J)$ and glue them together at the ``new'' marked points labeled by $k_1,\ldots,k_p$. Since we want to get the genus~$g$ in the end, we impose the condition $g_1+g_2+p-1=g$. We denote by 
$$
\DR_{g_1}(b_I,k_1,\ldots,k_p)\boxtimes\DR_{g_2}(-k_1,\ldots,-k_p,b_J)
$$
the resulting cycle in~$\oM_{g,m}$.

We have the following formula (see \cite{BSSZ12}):
\begin{align}\label{eq:splitting formula}
st_*q^* D^{Sym}_{(0,x_1|x_2,\infty)}=&\sum_{\substack{I\sqcup J=\{1,\ldots,n\}\\I,J\ne\emptyset\\\sum_{i\in I}a_i<0}}\sum_{\substack{g_1,g_2\ge 0\\p\ge 1\\g_1+g_2+p-1=g}}\sum_{\substack{k_1,\ldots,k_p\ge 1\\\sum_{i\in I}a_i+\sum_{i=1}^p k_i=0}}\frac{\prod_{i=1}^p k_i}{p!}\times\\
&\times\DR_{g_1}(0_{p_1},a_I,k_1,\ldots,k_p)\boxtimes\DR_{g_2}(-k_1,\dots,-k_p,0_{p_2},a_J).\notag
\end{align} 
Here the symbol $0_{p_i}$ means that the marked point $p_i$ has zero label.

Denote the class $st_*q^* D^{Sym}_{(0,x_1|x_2,\infty)}$ by $D_{x_1,x_2;a_1,\ldots,a_n}$. Since $D^{Sym}_{(0,x_1|x_2,\infty)}-D^{Sym}_{(0,x_2|x_1,\infty)}=0$ (\cite{LM00}), we have
\begin{gather}\label{eq:main equation}
D_{x_1,x_2;a_1,\ldots,a_n}-D_{x_2,x_1;a_1,\ldots,a_n}=0.
\end{gather}
The formulas~\eqref{eq:splitting formula} and~\eqref{eq:main equation} play the crucial role in the proof of Theorem~\ref{theorem:main}.


\subsection{Proof of Theorem~\ref{theorem:main}}\label{subsection:proof of the main theorem}

Since the class $\lambda_g$ vanishes on $\oM_{g,n}\backslash\M_{g,n}^{ct}$, we have
\begin{gather*}
\int_{\DR_{g_1}(0_{p_1},a_I,k_1,\ldots,k_p)\boxtimes\DR_{g_2}(-k_1,\dots,-k_p,0_{p_2},a_J)}\psi_{p_1}^{d_1}\psi_{p_2}^{d_2}\lambda_g c_{g,n+2}(e_{\alpha_1}\otimes e_{\beta_I}\otimes e_{\alpha_2}\otimes e_{\beta_J})=0,\text{ if $p\ge 2$}.
\end{gather*}
Here by $\psi_{p_i}$ we denote the psi-class corresponding to the marked point~$p_i$. If $p=1$, then from equation~\eqref{eq:gluing equation} it follows that
\begin{align*}
\int_{\DR_{g_1}(0_{p_1},a_I,k)\boxtimes\DR_{g_2}(-k,0_{p_2},a_J)}&\psi_{p_1}^{d_1}\psi_{p_2}^{d_2}\lambda_g c_{g,n+2}(e_{\alpha_1}\otimes e_{\beta_I}\otimes e_{\alpha_2}\otimes e_{\beta_J})=\\
=&\sum_{\mu,\nu}\eta^{\mu,\nu}\left(\int_{\DR_{g_1}(0_{p_1},a_I,k)}\psi_{p_1}^{d_1}\lambda_{g_1}c_{g_1,|I|+2}(e_{\alpha_1}\otimes e_{\beta_I}\otimes e_{\mu})\right)\times\\
&\times\left(\int_{\DR_{g_2}(-k,0_{p_2},a_J)}\psi_{p_2}^{d_2}\lambda_{g_2}c_{g_2,|J|+2}(e_{\nu}\otimes e_{\alpha_2}\otimes e_{\beta_J})\right).
\end{align*}
Using \eqref{eq:splitting formula} we get
\begin{multline*}
\sum_{g\ge 0}\sum_{n\ge 2}\frac{(-\hbar)^g}{n!}\sum_{\substack{a_1,\ldots,a_n\ne 0\\a_1+\ldots+a_n=0\\1\le\beta_1,\ldots,\beta_n\le N}}\left(\int_{D_{x_1,x_2;a_1,\ldots,a_n}}\lambda_g\psi_{p_1}^{d_1}\psi_{p_2}^{d_2}c_{g,n+2}(e_{\alpha_1}\otimes e_{\alpha_2}\otimes e_{\beta_1}\otimes\ldots\otimes e_{\beta_n})\right)\prod_{i=1}^n p_{a_i}^{\beta_i}=\\
=\sum_{\substack{k\ge 1\\\mu,\nu}}k\eta^{\mu,\nu}\der{g_{\alpha_1,d_1}}{p^{\mu}_k}\der{g_{\alpha_2,d_2}}{p^\nu_{-k}}.
\end{multline*}
In the same way we obtain
\begin{multline*}
-\sum_{g\ge 0}\sum_{n\ge 2}\frac{(-\hbar)^g}{n!}\sum_{\substack{a_1,\ldots,a_n\ne 0\\a_1+\ldots+a_n=0\\1\le\beta_1,\ldots,\beta_n\le N}}\left(\int_{D_{x_2,x_1;a_1,\ldots,a_n}}\lambda_g\psi_{p_2}^{d_2}\psi_{p_1}^{d_1}c_{g,n+2}(e_{\alpha_2}\otimes e_{\alpha_1}\otimes e_{\beta_1}\otimes\ldots\otimes e_{\beta_n})\right)\prod_{i=1}^n p_{a_i}^{\beta_i}=\\
=-\sum_{\substack{k\ge 1\\\mu,\nu}}k\eta^{\mu,\nu}\der{g_{\alpha_2,d_2}}{p^{\mu}_k}\der{g_{\alpha_1,d_1}}{p^\nu_{-k}}.
\end{multline*}
Summing these two expressions and using~\eqref{eq:main equation}, we obtain
$$
0=\sum_{\substack{k\ge 1\\\mu,\nu}}k\eta^{\mu,\nu}\der{g_{\alpha_1,d_1}}{p^{\mu}_k}\der{g_{\alpha_2,d_2}}{p^\nu_{-k}}-\sum_{\substack{k\ge 1\\\mu,\nu}}k\eta^{\mu,\nu}\der{g_{\alpha_2,d_2}}{p^{\mu}_k}\der{g_{\alpha_1,d_1}}{p^\nu_{-k}}=-i\{g_{\alpha_1,d_1},g_{\alpha_2,d_2}\}_\eta.
$$
By Lemmas~\ref{lemma:pullback} and~\ref{lemma:surjectivity}, this implies that the bracket $\{\og_{\alpha_1,d_1},\og_{\alpha_2,d_2}\}_{\eta\d_x}$ is a linear combination of the functionals $\int u^\beta dx$. Since $\{\og_{\alpha_1,d_1},\og_{\alpha_2,d_2}\}_{\eta\d_x}\in\hLambda_N^{[1]}$, we get $\{\og_{\alpha_1,d_1},\og_{\alpha_2,d_2}\}_{\eta\d_x}=0$. Theorem~\ref{theorem:main} is proved.


{
\appendix

\section{Properties of the map $T_0\colon\Lambda\to\cB_N^{pol}$}\label{appendix}

Here we prove Lemmas~\ref{lemma:surjectivity} and~\ref{lemma:pullback}. We also formulate two auxiliary lemmas: Lemma~\ref{lemma:variational} will be used in the proof of Lemma~\ref{lemma:pullback} and Lemma~\ref{lemma:Z map} is used in the proofs of Lemmas~\ref{lemma:string equation} and~\ref{lemma:derivative}.

\subsection{Auxiliary lemmas}\label{appendix:auxiliary lemmas}

\begin{lemma}\label{lemma:variational}
For any local functional $\oh\in\Lambda_N$, we have 
\begin{gather}\label{eq:variational}
T\left(\frac{\delta\oh}{\delta u^\alpha}\right)=\sum_{n\ne 0}\der{}{p^\alpha_n}(T_0(\oh))e^{-inx}+T_0\left(\der{\oh}{u^\alpha}\right).
\end{gather}
\end{lemma}
\begin{proof}
It is sufficient to prove the lemma for a local functional $\oh$ of the form $\oh=\int u^{\alpha_1}_{n_1}\ldots u^{\alpha_k}_{n_k}dx$, $k\ge 1$. Let $d:=n_1+\ldots+n_k$. We have
\begin{gather*}
\frac{\delta\oh}{\delta u^\alpha}=\sum_{j=1}^k\delta_{\alpha,\alpha_j}(-\d_x)^{n_j}\prod_{r\ne j}u^{\alpha_r}_{n_r}=\sum_{j=1}^k(-1)^{n_j}\delta_{\alpha,\alpha_j}\sum_{c_1+\ldots+\widehat{c_j}+\ldots+c_k=n_j}\frac{n_j!}{\prod_{r\ne j} c_r!}\prod_{r\ne j}u^{\alpha_r}_{n_r+c_r}.
\end{gather*}
We get
\begin{gather*}
T\left(\frac{\delta\oh}{\delta u^\alpha}\right)=i^d\sum_{j=1}^k\delta_{\alpha,\alpha_j}\sum_{a_1,\ldots,\widehat{a_j},\ldots,a_k\ne 0}\left(-\sum_{r\ne j}a_r\right)^{n_j}\left(\prod_{r\ne j}a_r^{n_r}\right)\left(\prod_{r\ne j}p^{\alpha_r}_{a_r}\right)e^{i x\sum_{r\ne j}a_r}.
\end{gather*}
On the other hand we have
\begin{align*}
&T_0(\oh)=i^d\sum_{\substack{a_1+\ldots+a_k=0\\a_1,\ldots,a_k\ne 0}}\left(\prod_{r=1}^k a_r^{n_r}\right)\left(\prod_{r=1}^k p_{a_r}^{\alpha_r}\right),\\
&T_0\left(\der{\oh}{u^\alpha}\right)=i^d\sum_{j=1}^k\delta_{\alpha,\alpha_j}\delta_{0,n_j}\sum_{\substack{a_1,\ldots,\widehat{a_j},\ldots,a_k\ne 0\\a_1+\ldots+\widehat{a_j}+\ldots+a_k=0}}\left(\prod_{r\ne j}a_r^{n_r}\right)\left(\prod_{r\ne j}p^{\alpha_r}_{a_r}\right).
\end{align*}
Now formula~\eqref{eq:variational} is clear.
\end{proof}

Let $\oh\in\Lambda_N$ be an arbitrary local functional. We want to give a formula for $T_0\left(\frac{\d\oh}{\d u^\alpha}\right)$. Before doing that we need to introduce some notations. 

Let us consider formal variables $p^\alpha_n$ for all integers $n$ and let $\cB'_N\subset\mbC[[p^\alpha_n]]_{n\in\mbZ}$ be the subalgebra that consists of power series of the form
$$
f=\sum_{k\ge 0}\sum_{\substack{1\le\alpha_1,\ldots,\alpha_k\le N\\n_1,\ldots,n_k\in\mbZ\\n_1+\ldots+n_k=0}}f_{\alpha_1,\ldots,\alpha_k}^{n_1,\ldots,n_k}p^{\alpha_1}_{n_1}p^{\alpha_2}_{n_2}\ldots p^{\alpha_k}_{n_k},
$$
where $f_{\alpha_1,\ldots,\alpha_k}^{n_1,\ldots,n_k}\in\mbC$. For any symmetric matrix $\eta=(\eta^{\alpha,\beta})_{1\le\alpha,\beta\le N}\in\Mat_{N,N}(\mbC)$ we endow the algebra $\cB'_N$ with the Poisson algebra structure by $\{p^\alpha_m,p^\beta_n\}_\eta:=i m \eta^{\alpha,\beta}\delta_{m+n,0}$. Denote by ${\cB'}_N^{pol}\subset\cB'_N$ the subspace that consists of power series of the form
$$
f=\sum_{k\ge 0}\sum_{\substack{1\le\alpha_1,\ldots,\alpha_k\le N\\n_1,\ldots,n_k\in\mbZ\\n_1+\ldots+n_k=0}}P_{\alpha_1,\ldots,\alpha_k}(n_1,\ldots,n_k)p^{\alpha_1}_{n_1}p^{\alpha_2}_{n_2}\ldots p^{\alpha_k}_{n_k},
$$
where $P_{\alpha_1,\ldots,\alpha_k}(z_1,\ldots,z_k)\in\mbC[z_1,\ldots,z_k]$ are some polynomials and the degrees of them in the power series $f$ are bounded from above.

Given a series 
$$
f=\sum_{k\ge 0}\sum_{\substack{1\le\alpha_1,\ldots,\alpha_k\le N\\n_1,\ldots,n_k\ne 0\\n_1+\ldots+n_k=0}}P_{\alpha_1,\ldots,\alpha_k}(n_1,\ldots,n_k)p^{\alpha_1}_{n_1}p^{\alpha_2}_{n_2}\ldots p^{\alpha_k}_{n_k}\in\cB_N^{pol},
$$
define a series $Z(f)\in{\cB'}^{pol}_N$ by
$$
Z(f):=\sum_{k\ge 0}\sum_{\substack{1\le\alpha_1,\ldots,\alpha_k\le N\\n_1,\ldots,n_k\in\mbZ\\n_1+\ldots+n_k=0}}P_{\alpha_1,\ldots,\alpha_k}(n_1,\ldots,n_k)p^{\alpha_1}_{n_1}p^{\alpha_2}_{n_2}\ldots p^{\alpha_k}_{n_k}.
$$
We have obtained a map $Z\colon\cB_N^{pol}\to{\cB'}_N^{pol}$.

\begin{lemma}\label{lemma:Z map}
For any local functional $\oh\in\Lambda_N$, we have
$$
T_0\left(\der{\oh}{u^\alpha}\right)=
\begin{cases}
1,&\text{if $\oh=\int u^\alpha dx$},\\
\left.\der{Z(T_0(\oh))}{p^\alpha_0}\right|_{p^*_0=0},&\text{otherwise}.
\end{cases}
$$
\end{lemma}
\begin{proof}
The cases $\oh=\int 1 dx$ and $\oh=\int u^\mu_d dx$ are obvious. It remains to check the lemma for $\oh=\int u^{\alpha_1}_{n_1}\ldots u^{\alpha_k}_{n_k}dx$, where $k\ge 2$. This can be done by an easy direct computation.
\end{proof}


\subsection{Proof of Lemma~\ref{lemma:surjectivity}}

Let us prove the surjectivity. The space $\cB_N^{pol}$ is spanned by the constants and the elements of the form
$$
s^{\alpha_1,\ldots,\alpha_k}_{d_1,\ldots,d_k}:=\sum_{\substack{a_1+\ldots+a_k=0\\a_1,\ldots,a_k\ne 0}}\left(\prod_{i=1}^k a_i^{d_i}\right)\left(\prod_{i=1}^k p^{\alpha_i}_{a_i}\right),\quad k\ge 2.
$$
It is easy to see that $(-i)^{\sum d_j} T_0\left(\int u^{\alpha_1}_{d_1}\ldots u^{\alpha_k}_{d_k} dx\right)=s^{\alpha_1,\ldots,\alpha_k}_{d_1,\ldots,d_k}$. Now the surjectivity is clear.

Let us find the kernel of the map $T_0$. Define $\Lambda_N^{[k,l]}$ to be the subspace of $\Lambda_N$ spanned by the local functionals of the form $\int u^{\alpha_1}_{d_1}\ldots u^{\alpha_l}_{d_l}dx$, where $\sum_{i=1}^l d_i=k$. Let $\cB_N^{pol;k,l}\subset\cB_N^{pol}$ be the subspace spanned by the elements of the form 
$$
\sum_{\substack{a_1+\ldots+a_k=0\\a_1,\ldots,a_k\ne 0}}P(a_1,\ldots,a_l)p^{\alpha_1}_{a_1}\ldots p^{\alpha_l}_{a_l},
$$
where $P(z_1,\ldots,z_l)\in\mbC[z_1,\ldots,z_l]$ are homogeneous polynomials of degree $k$. It is clear that~$T_0$ maps $\Lambda_N^{[k,l]}$ in~$\cB_N^{pol;k,l}$. Obviously, the kernel of the map $T_0\colon\bigoplus_{\substack{k\ge 0\\l\le 1}}\Lambda^{[k,l]}_N\to\bigoplus_{\substack{k\ge 0\\l\le 1}}\cB_N^{pol;k,l}$ is spanned by the local functionals $\int u^\alpha dx\in\Lambda_N^{[0,1]}$. It remains to prove that the map $T_0\colon\Lambda_N^{[k,l]}\to\cB_N^{pol;k,l}$ is injective for $l\ge 2$. Consider a local functional $\of=\int f dx\in\Lambda_N^{[k,l]}$ such that $T_0(\of)=0$. We have
$$
T(f)=\sum_{1\le\alpha_1,\ldots,\alpha_l\le N}\sum_{a_1,\ldots,a_l\ne 0}P_{\alpha_1,\ldots,\alpha_l}(a_1,\ldots,a_l)\left(\prod_{i=1}^lp^{\alpha_i}_{a_i}\right)e^{ix\sum a_j},
$$ 
where the polynomials $P_{\alpha_1,\ldots,\alpha_l}(z_1,\ldots,z_l)\in\mbC[z_1,\ldots,z_l]$ have degree $k$ and satisfy the property
$$
P_{\alpha_1,\ldots,\alpha_l}(z_1,\ldots,z_l)=P_{\alpha_{\sigma_1},\ldots,\alpha_{\sigma_l}}(z_{\sigma_1},\ldots,z_{\sigma_l}),
$$
for an arbitrary permutation $\sigma\in S_l$. Define the subset $H_l\subset\mbC^l$ by 
$$
H_l:=\left\{(z_1,\ldots,z_l)\in\mbC^l\left|
\begin{smallmatrix}
z_1,\ldots,z_l\in\mbZ_{\ne 0},\\
z_1+\ldots+z_l=0,\\
z_i\ne z_j.
\end{smallmatrix}
\right.
\right\}.
$$
From the fact, that $T_0(f)=0$, it follows that $P_{\alpha_1,\ldots,\alpha_l}|_{H_l}=0$. Since $l\ge 2$, we have $k\ge 1$ and
$$
P_{\alpha_1,\ldots,\alpha_l}(z_1,\ldots,z_l)=i(z_1+\ldots+z_l)Q_{\alpha_1,\ldots,\alpha_l}(z_1,\ldots,z_l),
$$
for some homogeneous polynomial $Q_{\alpha_1,\ldots,\alpha_l}(z_1,\ldots,z_l)$ of degree $k-1$. In the same way, as we proved the surjectivity of the map $T_0$, it is easy to show that there exists a differential polynomial $g\in\Lambda_N^{[k-1,l]}$ such that 
$$
T(g)=\sum_{1\le\alpha_1,\ldots,\alpha_l\le N}\sum_{a_1,\ldots,a_l\ne 0}Q_{\alpha_1,\ldots,\alpha_l}(a_1,\ldots,a_l)\left(\prod_{i=1}^lp^{\alpha_i}_{a_i}\right)e^{ix\sum a_j}.
$$ 
It is clear that $\d_x g=f$ and, therefore, $\of=\int f dx=0$. The lemma is proved.


\subsection{Proof of Lemma~\ref{lemma:pullback}}

We have to prove that $\{T_0(\oh_1),T_0(\oh_2)\}_\eta=T_0\left(\{\oh_1,\oh_2\}_{\eta\d_x}\right)$, for arbitrary two local functionals $\oh_1,\oh_2\in\Lambda_N$. We have the following chain of equations:
\begin{align*}
T_0\left(\{\oh_1,\oh_2\}_{\eta\d_x}\right)=&T_0\int\left(\sum_{\alpha,\beta}\frac{\delta\oh_1}{\delta u^\alpha}\eta^{\alpha,\beta}\d_x\frac{\delta\oh_2}{\delta u^\beta}\right)dx=\Coef_{e^{i0x}}\left(\sum_{\alpha,\beta}T\left(\frac{\delta\oh_1}{\delta u^\alpha}\right)\eta^{\alpha,\beta}\frac{d}{dx}T\left(\frac{\delta\oh_2}{\delta u^\beta}\right)\right)=\\
\stackrel{\text{by Lemma~\ref{lemma:variational}}}{=}&\sum_{\alpha,\beta}\sum_{n\ne 0}i n\eta^{\alpha,\beta}\der{T_0(\oh_1)}{p^\alpha_n}\der{T_0(\oh_2)}{p^\beta_{-n}}=\{T_0(\oh_1),T_0(\oh_2)\}_\eta.
\end{align*}
The lemma is proved.

}


\begin{thebibliography}{BCRR14}

\bibitem[ACG11]{ACG11} E. Arbarello, M. Cornalba, P. A. Griffiths. Geometry of algebraic curves. Volume II. With a contribution by Joseph Daniel Harris. Grundlehren der Mathematischen Wissenschaften [Fundamental Principles of Mathematical Sciences], 268. Springer, Heidelberg, 2011. xxx+963 pp.

\bibitem[BCRR14]{BCRR14} A. Brini, G. Carlet, S. Romano, P. Rossi. Rational reductions of the 2D-Toda hierarchy and mirror symmetry. {\it arXiv:1401.5725}.

\bibitem[BCR12]{BCR12} A. Brini, G. Carlet, P. Rossi. Integrable hierarchies and the mirror model of local $\CP1$. {\it Physica D: Nonlinear Phenomena} {\bf 241} (2012) 2156-2167.

\bibitem[Bur13]{Bur13} A. Buryak. Dubrovin-Zhang hierarchy for the Hodge integrals. {\it arXiv:1308.5716}.

\bibitem[Bur14a]{Bur14a} A. Buryak. Equivalence of the open KdV and the open Virasoro equations for the moduli space of Riemann surfaces with boundary. {\it arXiv:1409.3888}.

\bibitem[Bur14b]{Bur14b} A. Buryak. Open intersection numbers and a wave function of the KdV hierarchy. {\it arXiv:1409.7957}.

\bibitem[BPS12a]{BPS12a} A. Buryak, H. Posthuma, S. Shadrin. On deformations of quasi-Miura transformations and the Dubrovin-Zhang bracket. {\it Journal of Geometry and Physics} {\bf 62} (2012), no. 7, 1639-1651.

\bibitem[BPS12b]{BPS12b} A. Buryak, H. Posthuma, S. Shadrin. A polynomial bracket for the Dubrovin-Zhang hierarchies. {\it Journal of Differential Geometry} {\bf 92} (2012), no. 1, 153-185.

\bibitem[BSSZ12]{BSSZ12} A. Buryak, S. Shadrin, L. Spitz, D. Zvonkine. Integrals of psi-classes over double ramification cycles. To appear in {\it American Journal of Mathematics}, {\it arXiv:1211.5273}.

\bibitem[CDZ04]{CDZ04} G. Carlet, B. Dubrovin, Y. Zhang. The extended Toda hierarchy. {\it Moscow Mathematical Journal} {\bf 4} (2004) 313-332.

\bibitem[CvdL13]{CvdL13} G. Carlet, J. van de Leur. Hirota equations for the extended bigraded Toda hierarchy and the total descendent potential of $\CP1$ orbifolds. {\it Journal of Physics A: Mathematical and Theoretical} {\bf 46} (2013) 405205.

\bibitem[CMW12]{CMW12} R. Cavalieri, S. Marcus, J. Wise. Polynomial families of tautological classes on $\mathcal{M}_{g,n}^{rt}$. {\it Journal of Pure and Applied Algebra} {\bf 216} (2012), no. 4, 950-981.

\bibitem[DZ04]{DZ04} B. Dubrovin, Y. Zhang. Virasoro symmetries of the extended Toda hierarchy. {\it Communications in Mathematical Physics} {\bf 250} (2004), 161-193.

\bibitem[DZ05]{DZ05} B.~A.~Dubrovin, Y.~Zhang. Normal forms of hierarchies of integrable PDEs, Frobenius manifolds and Gromov-Witten invariants. A new 2005 version of {\it arXiv:math/0108160v1}, 295 pp.

\bibitem[Dub13]{Dub13} B. Dubrovin. Gromov-Witten invariants and integrable hierarchies of topological type. {\it arXiv:1312.0799}.

\bibitem[Eli07]{Eli07} Y. Eliashberg. Symplectic field theory and its applications. International Congress of Mathematicians. Vol. I, 217-246, Eur. Math. Soc., Zurich, 2007. 

\bibitem[EGH00]{EGH00} Y. Eliashberg, A. Givental, H. Hofer. Introduction to symplectic field theory. GAFA 2000 (Tel Aviv, 1999). Geom. Funct. Anal. 2000, Special Volume, Part II, 560-673.

\bibitem[FP05]{FP05} C. Faber, R. Pandharipande. Relative maps and tautological classes. {\it Journal of the European Mathematical Society} {\bf 7} (2005), no. 1, 13-49. 

\bibitem[FSZ10]{FSZ10} C.~Faber, S.~Shadrin, D.~Zvonkine. Tautological relations and the r-spin Witten conjecture. {\it Annales Scientifiques de l'\'Ecole Normale Sup\'erieure}~(4)~\textbf{43} (2010), no.~4, 621-658.

\bibitem[FR11]{FR11} O. Fabert, P. Rossi. String, dilaton, and divisor equation in symplectic field theory. {\it International Mathematical Research Notices} {\bf 2011}, no. 19, 4384-4404.

\bibitem[FJR13]{FJR13} H. Fan, T. Jarvis, and Y. Ruan. The Witten equation, mirror symmetry, and quantum singularity theory. {\it Annals of Mathematics} {\bf 178} (2013), no. 1, 1-106.

\bibitem[GJV11]{GJV11} I. P. Goulden, D. M. Jackson, R. Vakil. The moduli space of curves, double Hurwitz numbers, and Faber's intersection number conjecture. {\it Annals of Combinatorics} {\bf 15} (2011), no. 3, 381-436.

\bibitem[GV05]{GV05} T. Graber, R. Vakil. Relative virtual localization and vanishing of tautological classes on moduli spaces of curves. {\it Duke Mathematical Journal} {\bf 130} (2005), no. 1, 1-37. 

\bibitem[GZ12]{GZ12} S. Grushevsky, D. Zakharov. The double ramification cycle and the theta divisor. {\it Proceedings of the American Mathematical Society}~{\bf 142} (2014), no. 12, 4053-4064.

\bibitem[Hain13]{Hain13} R. Hain. Normal Functions and the Geometry of Moduli Spaces of Curves. Handbook of moduli. Vol. I, 527-578, Adv. Lect. Math. (ALM), 24, Int. Press, Somerville, MA, 2013. 

\bibitem[Kaz09]{Kaz09} M. Kazarian. KP hierarchy for Hodge integrals. {\it Advances in Mathematics} {\bf 221} (2009), no. 1, 1-21.

\bibitem[Kon92]{Kon92} M. Kontsevich. Intersection Theory on the Moduli Space of Curves and the Matrix Airy Function. {\it Communications in Mathematical Physics} {\bf 147} (1992), 1-23.

\bibitem[KM94]{KM94} M. Kontsevich, Yu. Manin. Gromov-Witten classes, quantum cohomology, and enumerative geometry. {\it Communications in Mathematical Physics} {\bf 164} (1994), no. 3, 525-562.

\bibitem[LM00]{LM00} A.~Losev, Y.~Manin. New moduli spaces of pointed curves and pencils of flat connections. {\it Michigan Mathematical Journal} {\bf 48} (2000), 443-472.

\bibitem[MW13]{MW13} S. Marcus, J. Wise. Stable maps to rational curves and the relative Jacobian. {\it arXiv:1310.5981}.

\bibitem[MST14]{MST14} T. Milanov, Y. Shen, H.-H. Tseng. Gromov--Witten theory of Fano orbifold curves and ADE-Toda Hierarchies. {\it arXiv:1401.5778}.

\bibitem[MT08]{MT08} T. Milanov, H.H. Tseng. The spaces of Laurent polynomials, $\mbP^1$-orbifolds, and integrable hierarchies. {\it Journal fur die Reine und Angewandte Mathematik} {\bf 622} (2008), 189-235.

\bibitem[Mum83]{Mum83} D. Mumford. Towards an enumerative geometry of the moduli space of curves. Arithmetic and geometry, Vol. II, 271-328, Progr. Math., 36, Birkhäuser Boston, Boston, MA, 1983. 

\bibitem[OP06]{OP06} A. Okounkov, R. Pandharipande. The equivariant Gromov-Witten theory of $\mbP^1$. {\it Annals of Mathematics} {\bf 163} (2006), no. 2, 561-605.

\bibitem[PST14]{PST14} R.~Pandharipande, J.~P.~Solomon, R.~J.~Tessler. Intersection theory on moduli of disks, open KdV and Virasoro. {\it arXiv:1409.2191}.

\bibitem[Ros10]{Ros10} P. Rossi. Integrable systems and holomorphic curves. Proceedings of the Gokova Geometry-Topology Conference 2009, 34-57, Int. Press, Somerville, MA, 2010.

\bibitem[Sha09]{Sha09} S. Shadrin. BCOV theory via Givental group action on cohomological field theories. {\it Moscow Mathematical Journal} {\bf 9} (2009), no. 2, 411-429.

\bibitem[Wit91]{Wit91} E. Witten. Two-dimensional gravity and intersection theory on moduli space. {\it Surveys in Differential
Geometry} {\bf 1} (1991), 243-310.

\bibitem[Wit93]{Wit93} E. Witten. Algebraic geometry associated with matrix models of two-dimensional gravity. In: Topological methods in modern mathematics (Stony Brook, NY, 1991), 235-269, Publish or Perish, Houston, TX (1993).

\end{thebibliography}
\end{document}